\documentclass[12pt]{iopart}
\usepackage{bm}
\usepackage{graphicx}
\usepackage{amssymb}
\usepackage{amsthm}
\newtheorem{theorem}{Theorem}

\newtheorem{lemma}{Lemma}
\theoremstyle{remark}
\newtheorem{remark}{Remark}
\usepackage{hyperref}

\usepackage{xcolor}
\definecolor{codegreen}{rgb}{0,0.6,0}
\definecolor{backcolour}{rgb}{0.95,0.95,0.92}

\usepackage{listings}
\lstdefinestyle{mystyle}{
    backgroundcolor=\color{backcolour},   
    commentstyle=\color{codegreen},
    keywordstyle=\color{magenta},
    basicstyle=\ttfamily\footnotesize,
}
\newcommand\ket[1]{\vert #1 \rangle}
\newcommand\bra[1]{\langle #1 \vert}
\newcommand\braket[2]{\langle #1 \vert #2 \rangle}
\renewcommand\Re[1]{\mathrm{Re}{#1}}

\newcommand\pdag{{\phantom\dag}}
\newcommand\half{{\textstyle\frac12}}
\begin{document}

\title[Symplectic diagonalization, non-Hermitian Hamiltonian, and symmetries]{Third
quantization for bosons: symplectic diagonalization, non-Hermitian Hamiltonian, and symmetries}

\author{Steven Kim and Fabian Hassler}

\address{JARA Institute for Quantum Information, RWTH Aachen University, 52056 Aachen, Germany}
\ead{steven.kim@rwth-aachen.de}

\begin{abstract}
Open quantum systems that interact with a Markovian environment can be
described by a Lindblad master equation.  The generator of time-translation is
given by a Liouvillian superoperator $\mathcal{L}$ acting on the density
matrix of the system. As the Fock space for a single bosonic mode is already
infinite-dimensional, the diagonalization of the Liouvillian has to be done on
the creation- and annihilation-superoperators, a process called `third
quantization'. We propose a method to solve the Liouvillian for quadratic
systems using a single symplectic transformation.  We show that the
non-Hermitian effective Hamiltonian of the system, next to incorporating the
dynamics of the system, is a  tool to analyze its symmetries.  As an example,
we use the effective Hamiltonian to formulate a
$\mathcal{PT}$-`symmetry' of an open system. We describe how the inclusion of
source terms allows us to obtain the cumulant generating function for observables
such as the photon current.
\end{abstract}

%
%
%
%
%

\section{Introduction}

Third quantization is a technique to solve quantum systems that are coupled to
a Markovian environment.  A Lindblad master equation describes their
time-evolution and can be compactly expressed by a Liouvillian superoperator
$\mathcal{L}$.  The technique was first proposed in~\cite{prosen:08} to solve
fermionic systems and was later extended for bosons \cite{prosen:10}. 
It builds on earlier algebraic approaches in Liouville space, see
e.g.~\cite{ban, arimitsu}. Third quantization has been used to
solve bosons on a lattice \cite{guo:17} and also nonlinear systems by
employing weak symmetries \cite{mcdonald:22}. 
The
method works directly on the annihilation- and creation-superoperators.
Starting from the stationary state, new ladder operators generate `excited
states' with a simple exponential time-evolution. 
Recently, the connection of this
approach to a covariance matrix approach of Gaussian states \cite{barthel:21} and the Keldysh
path-integral formulation of open systems has been highlighted
\cite{kamenev:23,clerk:23}. The identification of eigenmodes with a simple
exponential time-evolution allows one to separate fast from slow degrees of
freedom efficiently. This makes it possible to reduce problems with multiple time-scales
to the slow time-dynamics of the relevant mode at long-times,
see~\cite{hassler:23}.

In this paper, we outline a method to `third quantize' a quadratic bosonic
Liouvillian using a symplectic transformation. This method draws inspiration
from the classical mechanics of the saddle-point dynamic of the bosonic
Keldysh action \cite{kamenev:23}. The main difference to~\cite{prosen:10} is
that we diagonalize the Liouvillian in a single step. This allows us to
extend the method in a straightforward fashion to the task of counting. There,
the time-evolution is not trace-preserving due to the inclusion of source
terms.  We start by introducing a superoperator formalism and a
trace-preserving transformation.  We show how the effective non-Hermitian
Hamiltonian emerges that encodes all the information about the dynamics of the
system.  We briefly discuss exceptional points where the Liouvillian is not
diagonalizable. We provide an explicit formula for the time-evolution operator
in this case.  We proceed by classifying symmetries of a Liouvillian.  We
provide a general recipe to elevate a symmetry of a closed system acting on
the (effective) Hamiltonian to a symmetry of the open system acting on the
Liouvillian. We show how $\mathcal{PT}$-`symmetries' fit in this framework. We
conclude with two instructive examples where the method is applied to open
bosonic systems.

\section{Symplectic diagonalization of complex symmetric matrices}
We will map the task of solving the Lindblad equation onto the diagonalization
of a complex matrix using a symplectic transformation. In the applications, the
symplectic form arises due to the bosonic commutation relation. We prepare the
main theorem by two lemmas.

The first lemma is a restatement of the diagonalization of a complex matrix:
\begin{lemma}\label{lemma1}
Given the standard symplectic form  ($J^T = J^{-1} = - J$)
\begin{equation}\label{eq:j}
	J = \pmatrix{
		0 & I_n \cr
		-I_n & 0
	}\,,
\end{equation}
where $I_n$ is the identity matrix  of size $n$, a symmetric matrix $L \in
\mathbb{C}^{2n \times 2n}$, $L=L^T$, with $J^{-1} L$ diagonalizable, can be
brought into the normal form
\begin{equation}
	PLQ  
	=\pmatrix{
	0 & \Lambda \cr
	\Lambda & 0};
\end{equation}
where $\Lambda$ is a diagonal matrix with entries $\lambda_1, \dots, \lambda_n$
and $P, Q\in \mathrm{GL}(2n, \mathbb{C}) $ are (invertible)  transition matrices with
\begin{equation}
	PJQ = J\,.
\end{equation}
\end{lemma}
\begin{proof} 
The matrix $A = J^{-1}L$ can be diagonalized by a transition
matrix $Q$ and a diagonal matrix $D$ with $Q^{-1} A Q = D$. Note that $A = - J^{-1}
A^T J$ such that $A$ is similar to $-A^T$. Thus, the eigenvalues of $A$
come in pairs $\pm \lambda_i$.  Due to this, we can choose
\begin{equation}
	D = \pmatrix{ -\Lambda & 0 \cr 
	               0    & \Lambda}.
\end{equation}
With $P=  J Q^{-1} J^{-1}$, we have $ J^{-1} P L Q = D$ and the result
follows by
\begin{equation}
	J D = \pmatrix{0 & \Lambda \cr
	  \Lambda & 0}\,.
\end{equation}
\end{proof}
\begin{remark}
The column vectors $\bi q_i$ of $Q$ together with the eigenvalues $\lambda_i$
can be obtained by solving the generalized eigenvalue problem
\begin{equation}
  L \bi q_i = \lambda_i J \bi q_i\,.
\end{equation}
\end{remark}
\begin{remark}
If $A=J^{-1}L$ is not diagonalizable, the theorem can be adjusted such that
$\Lambda$ involves Jordan blocks and 
\begin{equation}
	D = \pmatrix{ -\Lambda & 0 \cr
	0   & \Lambda^T }\,.
\end{equation}
We give an example of this case below.
\end{remark}
The following technical lemma will be the central ingredient of the theorem:
\begin{lemma}\label{lemma2}
Given a invertible transition matrix $T \in \mathrm{GL}(2n, \mathbb{C}) $ and an
arbitrary matrix $A\in \mathbb{C}^{2n \times 2n}$ with 
\begin{equation}\label{eq:comm}
 A = T^T A T^{-1}\,,
\end{equation}
we can find $B \in \mathrm{GL}(2n, \mathbb{C}) $ such that $B^2 = T$ ($B$ is
a `square root' of $T$) and
\begin{equation}\label{eq:comm_sqrt}
 A = B^T A B^{-1}\,.
\end{equation}
\end{lemma}
\begin{proof}
From (\ref{eq:comm}), it follows by iteration that
\begin{equation}
 A = (T^k)^T A (T^k)^{-1}, \qquad k \in \mathbb{N}
\end{equation}
and also that $A = p(T)^T A\, p(T)^{-1}$ for an arbitrary polynomial $p(x)$.
The function $f(x) = \sqrt{x}$ evaluated on $T$ only depends on the value on
the spectrum of $T$, see~\cite{gantmacher:60} for details. As the
spectrum of $T$ does not involve $0$, we can choose a single-valued branch of
$f(x)$ and find a polynomial $p_f(x)$ that has the same value as $f(x)$ on the
spectrum of $T$. The matrix $B= p_f(T)$ fulfills the requirements $B^2 =T$ and
(\ref{eq:comm_sqrt}) of the lemma.  
\end{proof}

Now, we have gathered the necessary ingredients for our main theorem.
\begin{theorem}\label{theorem1}
Given the conditions of lemma~\ref{lemma1}, the symmetric matrix $L$
can be brought into the normal form 
\begin{equation}\label{eq:normal_form}
	S^TLS =\pmatrix{
	0 & \Lambda \cr
	\Lambda & 0}
\end{equation}
by a symplectic transition matrix $S$ fulfilling
\begin{equation}
S^TJS = J\,.	
\end{equation}
\end{theorem}
\begin{proof}
With lemma~\ref{lemma1}, the symmetry of $L$, and the skew-symmetry of $J$, it
follows that
\begin{equation}\label{eq:pmq}
PLQ = Q^TLP^T = \pmatrix{
	0 & \Lambda \cr
\Lambda & 0} \qquad \mathrm{and} \qquad P J Q =Q^T J P^T = J\,.  \label{eq:Atilde}
\end{equation} 
This implies that 
\begin{equation}\label{eq:comm_m}
	L = T^T L T^{-1}\qquad \mathrm{and} \qquad J = T^T J T^{-1}\,,
\end{equation}
where $T=Q(P^{-1})^{T}= Q P^{-T}$.  Using lemma~\ref{lemma2}, we obtain $B$
with $B^2 = T$ and $B$ fulfilling the properties of (\ref{eq:comm_sqrt}).
Setting $S =B^{-1} Q$, we obtain ($X$ is either $L$ or $J$)
\begin{equation}
	S^T X S  = Q^T B^{-T} X B^{-1} Q \stackrel{(\ref{eq:comm_m})}= Q^T
	(B^{-T})^{2}  X Q= Q^T
	T^{-T}   X Q = P X Q \,.
\end{equation}
The result follows by comparing $S^T L S = P L Q$ and $S^T J S = PJ Q$ with
(\ref{eq:pmq}).

\end{proof}
\begin{remark}\label{remark3}
The normal form is unique up to reordering of the $\lambda_i$ or replacing
$\lambda_i$ by its partner $-\lambda_i$. In particular, for each $i\leq n$,
the normal form is invariant under exchange $\bi s_i \mapsto -\bi s_{i+n}$,
$\bi s_{i+n} \mapsto  \bi s_i$ of the $i$-th and $(i+n)$-th column of $S$
while inverting the sign of $\lambda_i \mapsto - \lambda_i$. In the
application for open systems, the sign of $\lambda_i$ will determine the
stability of the system. We will discuss how to fix the sign ambiguity for the
elements of $\Lambda$ below.
\end{remark}
\begin{remark}
For an alternative proof using symplectic geometry induced by the form $\omega(\bi{x}, \bi{y})=\bi{x}^T
J\bi{y}$ see~\cite{laub:74}.
\end{remark}
\begin{remark}
Note that the presented theorem is different from the better-known Williamson's
theorem for the diagonalization of a real, symmetric, positive matrix with a
symplectic transition matrix, see e.g.~\cite{nicacio:21}.
\end{remark}
\section{Solving a bosonic Lindblad master equation}
\subsection{Superoperator formalism and non-Hermitian Hamiltonian}
Open quantum systems interact with their environment, e.g. via emission or
absorption of photons. For Markovian environments, the Lindblad master
equation governing the dynamics is given by
\begin{equation}
	\dot{\rho} = -i[\mathcal{H},\rho] +  \sum_\mu\Bigl[  \mathcal{J}_\mu \rho
		\mathcal{J}_\mu^\dag - \half
		(\mathcal{J}_\mu^\dag \mathcal{J}_\mu \rho + \rho \mathcal{J}_\mu^\dag
	\mathcal{J}_\mu)\Bigr]\,,
\end{equation}
which describes the time evolution of the density matrix $\rho$ of the system.
The Hamiltonian $\mathcal{H}$ describes the closed system evolution and the
Lindblad jump operators $\mathcal{J}_\mu$ incorporate the coupling to the
environment.  The equation states that the time-evolution of the density
matrix is generated by the Liouvillian superoperator $\mathcal{L}$ with
$\dot{\rho} = \mathcal{L}\rho$.  The general solution of this equation is
given by the time-evolution operator $\exp(\mathcal{L}t)$. In order to
understand the dynamics of the system, we need to obtain the eigenvalues
$\lambda$ of the Liouvillian $\mathcal{L}$. In particular, knowledge of the
eigenvalues allows us to identify instabilities of the system with $\Re(\lambda) >
0$, where the dynamics diverge exponentially, or dissipative phase transitions
when $\Re(\lambda) \rightarrow 0^-$ \cite{minganti:18}.

For exact solvability, we assume that the Hamiltonian is quadratic in the
bosonic ladder operators $a_i, i=1,\dots, m$  and that the bath operators are
linear in ladder operators acting on the Fock space $\mathcal{F}$ of $m$
bosonic modes.  This means that
\begin{eqnarray}
	\mathcal{H} = \sum_{i,j} \Bigl[a_i^\dag h_{ij} a_j + \half a_i \Delta_{ij} a_j
	+ \half a^\dag_i (\Delta^\dag)_{ij} a^\dag_j \Bigr] + \sum_i \Bigl(\alpha_i^* a^\pdag_i + \alpha_i a^\dag_i \Bigr), \\ 
	\mathcal{J}_\mu = \sum_i \Bigl( v_{\mu,i}a_i + w_{\mu,i}a^\dag_i \Bigr)
	+\beta_\mu.
\end{eqnarray}
The matrix $h\in \mathbb{C}^{m\times m}$ is Hermitian and it is possible to
choose $\Delta \in \mathbb{C}^{m\times m}$ as a symmetric matrix.  The vectors
$\bi{v}_\mu, \bi{w}_\mu \in \mathbb{C}^{m}$  describe the interaction with the
environment.  Emission of photons is described by $\bi{v}_\mu$  while
$\bi{w}_\mu$ relates to absorption processes. The coherent drive $
\alpha_i\in\mathbb{C}$ and the parameters $\beta_\mu\in \mathbb{C}$ couple
linearly to the bosonic ladder operators and thus displace  the vacuum, see
below.

Since the operators in the Lindblad equation act on both sides of the density
matrix, we go over to a superoperator formalism.  We define a set of
superoperators by $\mathcal{O}_c = (\mathcal{O}_+ + \mathcal{O}_-)/2$
and $\mathcal{O}_q = (\mathcal{O}_+ - \mathcal{O}_-)$ where
$\mathcal{O}_+ \rho = \mathcal{O}\rho$ and $\mathcal{O}_- \rho = \rho
\mathcal{O}$. The subscript `c' and `q' denoting classical and quantum are
derived from the language used in the Keldysh path-integral description
\cite{clerk:23,kamenev:05}.  The relevant commutation relations are
$[a^\pdag_{c,i} , a_{q,j}^\dag] = [a^\pdag_{q,i} , a_{c,j}^\dag] =
\delta_{ij}$ while the remaining commutators vanish. The superoperators act
on the `doubled' Fock-space $\mathcal{F}^2 = \mathcal{F}_+ \otimes
\mathcal{F}_-$  in which the density matrix is vectorized such that
$\mathcal{L}$ acts as a linear map, see below. 

Using this definition and defining the basis $\bi{b} =(\bi b_c, \bi b_q)^T =
(\bi{a}^\pdag_c, \bi{a}_c^\dag, \bi{a}_q^\dag, -\bi{a}^\pdag_q)^T$, we can
express the Liouvillian by a complex symmetric matrix $L \in
\mathbb{C}^{2n\times 2n}$ with $n=2m$.  The Liouvillian reads $\mathcal{L} =
\frac12 \bi{b}^T L \bi{b} + \bm{\eta} \cdot \bi b_q+ L_0$ with
\begin{eqnarray} \label{eq:liouvillian}
	L &= \pmatrix{0 & -i H^T_{\mathrm{eff}} Z \cr 
	-i Z H_\mathrm{eff}   &  N}\\ \bm{\eta} &= \pmatrix{
\bi z \cr \bi z^* }, \quad \bi z = -i \bm\alpha + \frac{1}{2}\sum_\mu
\pmatrix{\beta^*_\mu \bi w_\mu - \beta_\mu \bi v^*_\mu}  \in
\mathbb{C}^m\,,
\end{eqnarray}
and $L_0 =\frac12 \tr(V-W)$.  We define the `dissipation matrices' $V$, $W$,
and $U \in \mathbb{C}^{m\times m}$ by the sum of dyadic products
\begin{equation}
V= V^\dag = \sum_\mu \bi{v}_\mu  \bi{v}_\mu^\dag , \quad
W  = W^\dag = \sum_\mu \bi{w}_\mu  \bi{w}_\mu^\dag, \quad
U = -\sum_\mu \bi{v}_\mu  \bi{w}_\mu^\dag.
\end{equation}

We note that $L$ has a block  structure. The zero matrix in the upper-left
corner derives from the trace preservation of the Lindblad equation
\cite{prosen:10,kamenev:23}.  This block would connect the `c'-operators to
each other.  Since these terms are absent, in each term of the Liouvillian
there is at least one `q'-operator and the trace preservation follows from
$\Tr(\mathcal{O}_q \rho) = \Tr(\mathcal{O}\rho - \rho \mathcal{O}) = 0$.  It
is the trace preservation that allows the diagonalization with the method of
third quantization. In the superoperator formalism, the trace over the Fock
space $\mathcal{F}$ is a distinguished left-eigenvector $\bra{ 0 }\cdots
\equiv \Tr(\cdots)$ from which the rest of the spectrum can be obtained by
ladder operators, see below.

Due to the zero matrix in the upper-left block, the lower-right block matrix
$N \in \mathbb{C}^{2m\times 2m}$, given by
\begin{equation} \label{eq:noise}
N = \frac12 \pmatrix{U^\dag+U^* & W+V^* \cr
W^*+V & U^T+U },
\end{equation}
does not influence the spectrum $\lambda$ of the Liouvillian. In fact, it
encodes the noise, see \cite{kamenev:23}. The off-diagonal blocks encode the
dynamics of the system. To connect to the literature about non-Hermitian
Hamiltonians, we have split off the matrix
\begin{equation}
Z = \pmatrix{I_{m} & 0 \cr
0 & -I_{m}}
\end{equation}
which encodes the bosonic commutation relations $[a_i,a^\dag_j]=\delta_{ij}$
\cite{colpa:78}.  The remaining part is the non-Hermitian Hamiltonian
\begin{equation} \label{eq:Heff}
H_\mathrm{eff} = \pmatrix{h + \frac{i}{2}(W-V^*) & \Delta^* -\frac{i}{2} (U^\dag-U^*) \cr
\Delta + \frac{i}{2}(U^T - U) & h^* - \frac{i}{2}(W^*-V) }.
\end{equation}
It can be decomposed as $H_\mathrm{eff} = H + \frac{i}{2}\Gamma$, with $H$ and
$\Gamma$ Hermitian. The first part $H$ is the Hamiltonian of the closed system
given by $\mathcal{H} = \frac12 {\bi{b}^\prime}^\dagger H \bi{b}^\prime$ with
$\bi{b}^\prime = (\bi{a}, \bi{a}^\dagger)^T$.  The second part $\Gamma$
encodes the dissipative dynamics due to the Lindblad jump operators
$\mathcal{J}_\mu$.  Non-Hermitian Hamiltonians which encode only the dynamical
part of the Lindblad time-evolution, have recently received a lot of interest
in the context of non-Hermitian topology \cite{bergholtz:21}.

We continue by showing how to diagonalize the Liouvillian $L$. As we are in
phase-space, it is important to conserve the bosonic structure given by the
commutation relation
\begin{equation}
J_{ij} = \left[
\bi{b}_i
,
\bi{b}_j
\right]
\end{equation}
where $J$ is the standard symplectic form of (\ref{eq:j}). Thus, $L$ and
$J$ are in the form of theorem~\ref{theorem1} and we can find a symplectic
transformation $S$ that brings $L$ in the normal form of (\ref{eq:normal_form})
with the generalized eigenvalues $\lambda_i$ forming the diagonal matrix
$\Lambda$.

\subsection{Fixing the normal form}

As the eigenvalues $\lambda_i$ determine the dynamics and the stability of the
system, we need to fix the sign ambiguity $\pm \lambda_i$.  One of the columns
$\bi s_i$ and $\bi s_{i+n}$, $i\leq n$, of $S$ is of the form
$(\tilde{\bi{s}}_i,0\cdots0)^T$, i.e. ending in $n$ zeros, see
\cite{prosen:10}. If $\bi s_{i+n}$ ends in the zeros, we exchange $\bi s_i
\mapsto -\bi s_{i+n}$, $\bi s_{i+n} \mapsto \bi s_i$ while inverting the sign
of $\lambda_i \mapsto - \lambda_i$, according to remark~\ref{remark3}.  Note
that the normalization of the eigenvectors is given by $\bi s(-\lambda_i)J\bi
s(\lambda_j)=\delta_{ij}$. In this way, we fix the signs of $\lambda_i$ while
bringing the symplectic matrix in the block form
\begin{equation}\label{eq:s_block}
	S = \pmatrix{ S_1& S_2 \cr 0 & S_3}\,.
\end{equation}
The condition $S^TJ S =  J$ implies that $S_3^T =S_1^{-1}$ and $S_2^T S_3 =
S_3^T S_2$.  In the following, we will always assume that the normal form is
\emph{fixed} in this way. In the following, we call the (generalized)
eigenvalues $\lambda_i$ the (symplectic) spectrum of the Liouvillian $L$.

Having made sure that the columns $\bi s_i$, $i\leq n$, end in 0, the
generalized eigenvalue problem $L \bi s_i = \lambda_i J\bi s_i$ is equivalent
to
\begin{equation}\label{eq:heff_ev}
H_\mathrm{eff} \tilde{\bi{s}}_i = i\lambda_i Z \tilde{\bi{s}}_i \,,
\end{equation}
i.e. the characteristic values $\lambda_i$ are, up to multiplication by $i$
which corresponds to a rotation by $\pi/2$ in the complex plane, the
eigenvalues of the non-Hermitian effective Hamiltonian $H_\mathrm{eff}$. While
we promote the diagonalization with the symplectic transformation of
theorem~\ref{theorem1} with a subsequent fixing of the normal form, there is
the alternative approach of \cite{prosen:10}. They proceed by solving
(\ref{eq:heff_ev}) followed up by a solution of a Lyapunov equation to obtain
the full eigenvectors $\bi s_i$.

The transformation to the normal form introduces new bosonic operators
$u_i$ and $v_i$ defined by
\begin{equation}
	\bi{b} = S  \pmatrix{
		u_1- \eta_1'/\lambda_1\cr \vdots\cr u_n-\eta_n'/\lambda_n \cr v_1 \cr \vdots\cr v_n
	}\,, 
\end{equation}
which implies $\bi b_c = S_1 (\bi u - \Lambda^{-1}\bm{\eta}') + S_2 \bi v$ and
$\bi b_q = S_3 \bi v$. The linear terms $\bm{\eta}$ shift the vacuum as
parameterized by the transformed vector $\bm{\eta}'= S_3^T \bm{\eta} =
S_1^{-1} \bm{\eta}$.

Due to the symplectic nature of the transform, the operators obey $[u_i,
v_j]=\delta_{ij}$ while all other commutators vanish. Thus, $u_i$ are
annihilation  and $v_i$ creation operators.  Since the Liouvillian $L$ is not
a Hermitian operator, we do not have the property that $u_i$ is the adjoint of
$v_i$.  Importantly, the new creation operators $v_i$ are solely given by a
linear combination of $\bi{a}_{q}$ and $\bi{a}_{q}^\dag$. Because of this, the
$v_i$ take both the role of creation operators as well as the preservation of
the trace, see the discussion above.

Written with the new operators $u_i$ and $v_i$, the Liouvillian assumes the form
\begin{equation}
	\mathcal{L} = \sum_i \lambda_i v_i u_i   
\end{equation}
where the additional constant $\frac12 \sum_i \lambda_i + L_0$ vanishes in
accordance with the trace preservation.\footnote{We note that $\sum_i \lambda_i
= \tr(-iZH_\mathrm{eff})$ which evaluates to $ \tr(W-V) = -2L_0$ using
(\ref{eq:Heff}). }

\subsection{Fock space of the superoperator formalism}

The ladder operators $v_i, u_i$ introduce a Fock-like structure. In
particular, the operator $v_i u_i$ is a number operator. Given any eigenstate
$\ket{n_i}$ the eigenvalues can be increased (decreased) by $1$ by the
application of $v_i$ ($u_i$) due to the ladder structure $[v_i u_i, v_i ]= v_i$
($[v_i u_i, u_i]=-u_i$). To obtain the spectrum from the ladder
property, we need a distinguished state from which the rest of the states can
be obtained by application of the creation operators.  In the standard
discussion of the algebraic solution of the quantum harmonic oscillator the
vacuum plays this role. Here, the trace $\bra{0}$ plays the
distinguished role as it is a left eigenvector to the eigenvalue $0$, see
\cite{prosen:10}.

Starting from the state $\bra{0}$ with $\bra{0} v_i u_i = 0$, for $i\leq n$,
we can obtain the remaining state by application of $u_i$ to the left. Note
that the annihilation operator $u_i$ acts as a creation operator when acting to
the left. In this way, we obtain the remaining states, $n_i \in \mathbb{N}_0$,
\begin{equation}
	\bra{n_1, n_2, \dots } =  \bra{0}
	\prod_{i=1}^n\frac{u_i^{n_i}}{\sqrt{n_i!}}
\end{equation}
which are left eigenstates of $v_i u_i$ to the eigenvalue $n_i$. Because of
this, the spectrum of the Liouvillian $\mathcal{L}$ is given by
\begin{equation}
	\lambda = \sum_{i=1}^n \lambda_i n_i\,.
\end{equation}

To every left eigenvector $\bra{n_1, n_2, \dots }$ there is a corresponding
right eigenvector
\begin{equation}
  \ket{n_1, n_2, \dots } = 
  \prod_{i=1}^n\frac{v_i^{n_i}}{\sqrt{n_i!}} \ket{0}\,,
\end{equation}
which is not its Hermitian adjoint.  Importantly, the right eigenvector
$\ket{0}$ to $\lambda=0$ corresponds to the stationary density matrix $\rho_s
\equiv \ket{0}$, as  $\mathcal{L} \rho_{s} = \dot{\rho_{s}} = 0$.  The ladder
operators act on the eigenstates $\ket{n_1, n_2, \dots } = \bigotimes_i
\ket{n_i}$ according to $u_i \ket{n_i} = \sqrt{n_i} \ket{ n_i-1}$ and $v_i
\ket{n_i} = \sqrt{n_i+1} \ket{n_i+1}$ \footnote{It would be also possible to
chose a different `normalization' of the operators, e.g. $u_i \ket{n_i} =
\ket{n_i-1}$ and $v_i \ket{n_i} = (n_i+1) \ket{n_i+1}$. With the present
choice, the matrices implementing $u_i$ and $v_i$ in the Fock-basis are
adjoints of each other.}.  The action on the left eigenstates is given by
$\bra{n_i}u_i = \bra{n_i+1}\sqrt{n_i+1}$, $\bra{n_i}v_i =
\bra{n_i-1}\sqrt{n_i}$.  These states form a complete bi-orthonormal set with
$\braket{n_1', n_2', \dots }{n_1, n_2, \dots } = \Pi_i \delta_{n_i', n_i}$ and
$I_{\mathcal{F}^2} = \sum_{n_1, n_2, \dots} \ket{n_1, n_2, \dots }\bra{n_1,
n_2, \dots }$.

Using the completeness, the time evolution operator of the Lindblad equation
equals
\begin{equation}\label{eq:time_evolv}
\exp(\mathcal{L}t) = \sum_{n_1, n_2, \dots} e^{\lambda t} \ket{n_1, n_2, \dots
}\bra{n_1, n_2, \dots }
\end{equation}
with $\lambda = \sum_i \lambda_i n_i$.  This means that states with large
$\lambda$ (large $n_i$ or $\lambda_i$) decay faster than others. This insight
makes it possible to adiabatically eliminate the fast modes, see \cite{hassler:23}
for an example.

\subsection{Jordan normal-form}

Generally, it might happen that the matrix $J^{-1}L$ is not diagonalizable,
meaning that it can only be brought into its Jordan normal-form.  For the
following, we will discuss the simplest case for a system containing a single
boson where the eigenvalues and eigenvectors coalesce. In this case,
the Liouvillian equals $\mathcal{L} = \mu (v_1u_1 + v_2u_2) + \nu v_2 u_1$,
where the first part is diagonal and $v_2u_1$ is nilpotent.  The following
procedure can be straightforwardly extended to more general cases involving
larger Jordan blocks.

While the algebraic eigenvalues of $\mathcal{L}$ stays the same, the
geometrical spectrum that determines the dynamics of the system changes.  If
$\mathcal{L}$ would be diagonalizable with degenerate eigenvalues, there would
be a degeneracy of $n+1$ for each eigenvalue $\lambda = \mu (n_1 + n_2)$ with
$n_1 + n_2 = n$.  Because a complete decomposition of $\mathcal{L}$ by
eigenstates is not possible anymore, additional generalized eigenstates are
needed. Acting with the Liouvillian on the Fock-states, we obtain
\begin{equation}
	\mathcal{L} \ket{n_1, n_2} = \mu (n_1 + n_2)\ket{n_1, n_2}+ \nu \sqrt{n_1(n_2+1)}
	\ket{n_1-1, n_2+1} \,.
\end{equation}
We see that, for fixed $n = n_1 +n_2$, only the state $\ket{0,n}$ with $n_1=0$
is an eigenstate of the Liouvillian.  The $n$ remaining states $\ket{n_1, n
-n_1}$ with $n_1 = 1,\dots, n$ are generalized eigenstates.  The time
evolution operator reads
\begin{equation}\label{eq:time_evolv_jordan}
	e^{\mathcal{L}t} \!=\!\!\! \sum_{n_1, n_2 \geq 0\atop 0\leq m \leq n_1}  
	\sqrt{{n_1 \choose m} {n_2 + m \choose m}} \,(\nu t)^m e^{\mu (n_1 +
	n_2) t} \ket{n_1-m, n_2+m}\bra{n_1, n_2}\,,
\end{equation}
where $m$ runs over the set of generalized eigenstates to a fixed $n$. Note
that the time evolution is not a simple sum of exponential terms.  This is due
to the fact that the eigenstates do not form a complete basis anymore and
generalized eigenstates are needed that change the dynamics because of the
nilpotent operator $v_2u_1$.

\subsection{Counting statistics}

When solving the Lindblad equation, not only the stationary state $\rho_s$ and
the time evolution $\exp(\mathcal{L} t)$ are of interest. Often, one is
interested in the statistics of some observables $\mathcal{O}$, such as the
photon current.  We can access the statistics by including a source term in the
Liouvillian superoperator \cite{arndt:21, flindt:19}
\begin{equation}
\mathcal{L}(s) = \mathcal{L} + s\mathcal{O} 
\end{equation}
with $s\in \mathbb{C}$.  For the solvability of the model, the observable
$\mathcal{O}$ should again be at most of second order in the bosonic
operators $a_i$ and $a_i^\dag$.

A common example of an observable is the photon current $\mathcal{I}$. The
emission of a single photon is given by $\gamma a\rho a^\dag$ with $\gamma$
the emission rate.  Correspondingly, the photon current is given by
$\mathcal{I} = \gamma a_+ a^\dag_-$ \cite{arndt:21, padurariu:12, arndt:19}. Note that
the fact that the annihilation (creation) operator has the subscript $+$ ($-$)
makes sure that the operators are normal-ordered after taking the trace.
This is the conventional prescription for photon counting, see~\cite{mandel_wolf}.

We can obtain the generating function by diagonalizing
$\mathcal{L}(s)$ which yields
\begin{equation}
\mathcal{L}(s) = \sum_i  \lambda_i(s) v_i(s) u_i(s) +  \half\sum_i \lambda_i(s) + L_0.
\end{equation}
Note that in general, the counting field introduces a block $C(s)$ in the
upper-left corner of $L$ of (\ref{eq:liouvillian}). This is due to the fact
that the source term violates the trace preservation of the time evolution. In
particular, trace preservation only demands that $C(0 ) =0$. Because of this,
the eigenvalues $\lambda_i(s)$ and the operators $v_i(s), u_i(s)$ are chosen
such that they are analytical functions of  $s$. This remark is important as
at finite $s$ trace preservation cannot be used anymore to select the correct
eigenvalue among each of the $\pm \lambda_i$ pairs. Note that for this
case, it is necessary that we use the method based on theorem~\ref{theorem1}
as the alternative approach of \cite{prosen:10} relies on $C(s)=0$.

The factorial cumulant generating function $\mathcal{G}(s)$ for a measurement
time $\tau$ is defined by
\begin{equation}
	e^{\mathcal{G}(s) \tau} =  \bra{0} e^{\mathcal{L}(s) \tau} \ket{0} =
	\Tr\bigl(e^{\mathcal{L}(s) \tau} \rho_s \bigr) = \langle
	e^{ (\mathcal{L} + s \mathcal{O}) \tau}\rangle\,,
\end{equation}
where $\langle \cdot \rangle$ denotes the expectation value in the stationary
state.  The derivatives with respect to $s$ evaluated at $s=0$ produce the
factorial cumulants \cite{flindt:19,beenakker:98,flindt:11}
\begin{equation}
\langle \!\langle \mathcal{O}^k \rangle \!\rangle_F = \langle\! \langle
\mathcal{O}(\mathcal{O}-1)\cdots(\mathcal{O}-k+1)
\rangle \!\rangle = \left.\frac{d^k}{d s^k} \mathcal{G}(s) \right\vert_{s=0}.
\end{equation}
Note that the conventional cumulants can be obtained by replacing $s \mapsto
\exp(i\chi)-1$ and taking the derivative with respect to $i\chi$
\cite{flindt:11}. In the limit of long measurement time with $-\Re(\lambda_i)
\tau \gg 1$, only the term with $n_i=0$  contributes in (\ref{eq:time_evolv}).
The factorial cumulant generating function is then given by
\begin{equation}
  \mathcal{G}(s) =
\bra{0} \mathcal{L}(s) \ket{0}= \half \sum_i \lambda_i(s) + L_0 
= \half \sum_i \bigl[\lambda_i(s) -\lambda_i(0) \bigr]\,.
\end{equation}

\section{Symmetries}

\begin{figure}
  \centering
  \includegraphics[scale=0.7]{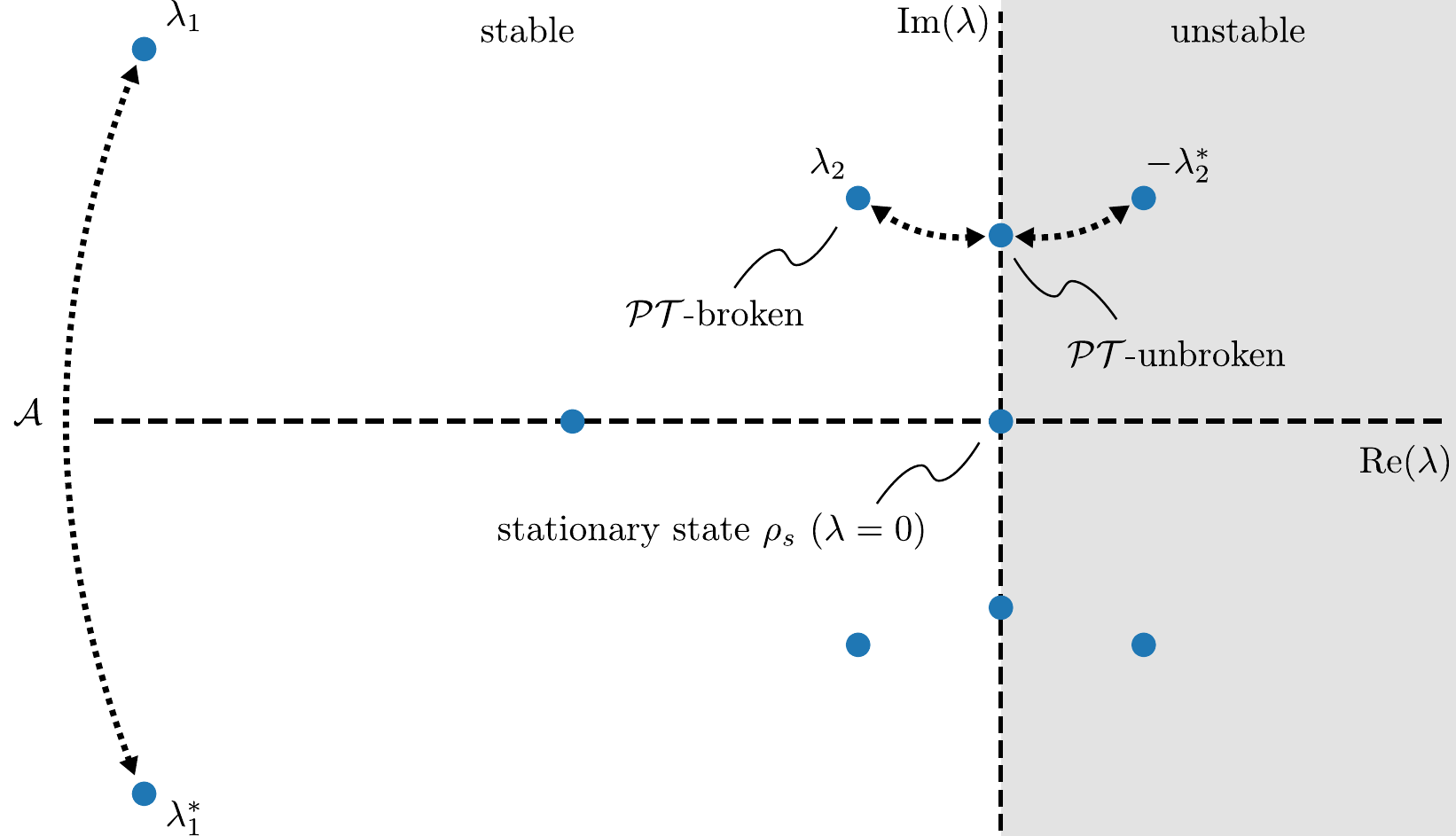}
  \caption{%
Sketch of the spectrum of a Liouvillian showing its key characteristics. The
eigenvector $\ket{0}$ to the eigenvalue $0$ corresponds to the stationary
state $\rho_s$. States with   eigenvalue $\Re(\lambda) < 0$ decay
exponentially and relax to the stationary state. Eigenvalues with
$\Re(\lambda) > 0$ (grey shaded region) indicate that the system is unstable
and the time-evolution $e^{\mathcal{L}t}$ diverges exponentially.  Due to the
hermitcity of the density-matrix $\rho(t)$ encoded in the adjoint operator
$\mathcal{A}$ with $[\mathcal{A}, \mathcal{L}]=0$, the eigenvalues of the
Liouvillian come in complex conjugate pairs. Thus, the spectrum is mirrored
along the real axis.  Similarly, for a $\mathcal{PT}$-`symmetry', the eigenvalues
come in pairs mirrored along the imaginary axis. If all the eigenvalues are on
the imaginary axis (corresponding to unitary dynamics), the
$\mathcal{PT}$-`symmetry' is called unbroken. In the opposite case with a broken
$\mathcal{PT}$-`symmetry', the system is unstable as one of the pair of
eigenvalues has a positive real part.
}\label{fig:spectrum}
\end{figure}

Symmetries are a powerful tool to get deeper insights.  For Hamiltonian
systems, the presence of a symmetry demands that the Hamiltonian $H$ commutes
with a  symmetry operator $\Omega$.  First attempts to generalize  symmetries
to the Lindblad equation have been studied in~\cite{albert:14, buca:12}.  In
the following, we approach this problem by using the superoperator formalism
and the structure of the Lindblad equation that has been introduced above.

\subsection{Hermiticity of the density matrix}

An admissible quantum state $\rho$ is Hermitian with  $\rho = \rho^\dag$. In
particular, this property has to be preserved under time-evolution. Note that
hermiticity is not a real symmetry as it cannot be broken. Still, as we will
see below, it introduces a structure that is akin to a real symmetry. As such,
it is instructive to first discuss the restriction that the  preservation of
hermiticity imposes on the Liouvillian $\mathcal{L}$.

In an orthonormal basis $\ket{n} = \ket{n_1, n_2, \dots}$, the state $\rho$
is represented by
\begin{equation}
\rho = \sum_{m,n} \rho_{mn}\vert m\rangle\langle n \vert 
\end{equation}
with $\rho_{mn} = \rho_{nm}^*$ due to hermiticity.  The superoperator
formalism  corresponds to a vectorization of the density matrix.  The
vectorized density matrix  reads
\begin{equation}
\vert \rho \rangle = \sum_{m,n} \rho_{mn} \ket{ m, n }
\end{equation}
where $\ket{m,n}= \ket{m}\otimes\ket{n}$, the factors corresponding to the $+$
and $-$ sectors of $\mathcal{F}^2$.

We want to find an operation $\mathcal{A}$ that corresponds to taking the
conventional adjoint $\rho^\dag$ of the density-matrix $\rho$.  This adjoint operation $\mathcal{A}$ is an anti-linear operator with $\mathcal{A}\,
z \ket{m,n} = z^* \ket{n,m}$. In the product basis, it is given by
\begin{equation}
\mathcal{A} = \sum_{m,n} \ket{n,m} \bra{m,n} K
\end{equation}
where $K$ is complex conjugation and $\mathcal{A}^2 = I_{\mathcal{F}^2}$. The hermiticity of
$\rho$ then corresponds to the relation $\mathcal{A} \ket{\rho} = \ket{\rho}$
on the vectorized density matrix.  The fact that hermiticity is preserved under
time-evolution demands that $[\mathcal{A}, \mathcal{L}] = 0$.

To see the restriction that the hermiticity puts upon the Liouvillian, we also
investigate the action of $\mathcal{A}$ on the superoperators.  This means
that (for notational simplicity, we suppress the subscript $i$ in the following)
\begin{equation}
	\mathcal{A}\ket{m,n}
	=  \frac{(\mathcal{A} a_+^\dag \mathcal{A})^m}
	{\sqrt{m!}}\frac{(\mathcal{A}a_-\mathcal{A})^n}{\sqrt{n!}}   \mathcal{A}\ket{0} 
	\stackrel{!}=  \frac{(a_-)^m}{\sqrt{m!}}\frac{(a_+^\dag)^n}{\sqrt{n!}} \ket{0}
 = \ket{n, m}.
\end{equation}
The stationary state $\ket{0}$ is Hermitian. Thus, we have $\mathcal{A}\ket{0}
= \ket{0}$ from which $\mathcal{A} a_+^\dag \mathcal{A} = a_-$ and
$\mathcal{A}a_-\mathcal{A} = a_+^\dag$ follows.  The action on the rotated
superoperators $a_c$ and $a_q$ is given by $\mathcal{A}a_c \mathcal{A} =
a_c^\dag$ and $\mathcal{A}a_q\mathcal{A} = -a_q^\dag$.  In a first quantized
version in the basis $\bi{b}=(\bi{a}^\pdag_c, \bi{a}_c^\dag, \bi{a}_q^\dag,
-\bi{a}^\pdag_q)^T$, this operation corresponds to 
\begin{equation}
A = \pmatrix{ X & 0 \cr
0 & X } K 
\qquad \mathrm{where} \qquad
X = \pmatrix{ 0 & I_m \cr
I_m& 0 }.
\end{equation}
The hermiticity constraint  requires that $A L A = L$ and thus $X
H_\mathrm{eff}^* X =H_\mathrm{eff}$ and $X N^* X =N$. As the $H_\mathrm{eff}$
and $N$ in (\ref{eq:Heff}) and (\ref{eq:noise}) derive from a Lindblad
equation, they obey this symmetry by construction. In particular, the
symplectic structure is also unchanged under hermiticity with  $A JA =J$.

The fact that the density matrix is hermitian also explains the symmetry under
mirroring along the real line in the spectrum of the Liouvillian, see
figure~\ref{fig:spectrum}. Given a generalized eigenvector $\bi{s}_i$ to the
eigenvalue $\lambda_i$ of the Liouvillian $L \bi s_i = \lambda_i J \bi s_i$,
the vector $\bi s'_i= A \bi s_i$ is an eigenvector to the eigenvalue
$\lambda_i^*$. This fact can be shown by straightforward calculation
\begin{equation}\label{eq:a_act}
L \bi s'_i = L A \bi s_i  = A L \bi s_i =  \lambda_i^* AJA \bi{s}'_i 
 =\lambda_i^* J \bi{s}'_i\,. 
\end{equation}
Due to this, the eigenvalues $\lambda_i$ are either real or come in complex
conjugate pairs. The eigenvectors $\bi s_i$ to the real
eigenvalues can be chosen such as to obey the constraint $A \bi s_i =  \bi
s_i$.

\subsection{Symmetries of  the Lindblad master equation}

For unitary dynamics, a Hamiltonian $\mathcal{H} = \frac12 \bi{b}'^\dag
H\bi{b}'$ in the basis $\bi{b}' =(\bi a, \bi a^\dag)^T$ may contain a symmetry.
This is implemented by a unitary matrix $\Omega$ \footnote{If $\Omega$ is
anti-unitary, we choose $\Omega X K$ as a unitary symmetry (due to
hermiticity).} which transforms $\bi b' \mapsto \Omega \bi b'$. In order to be
a symmetry, it has to fulfill $
\Omega^\dag H \Omega=H $, moreover $ \Omega^\dag  \,iZ\,
\Omega =iZ$, to preserve the commutation relation of the bosons, and $X
\Omega^* X = \Omega$, to keep $b_i^\dag$ the adjoint of $b_i$. 

For any unitary symmetry $\Omega$ of a closed system, we define a symmetry of
the full Liouvillian $L$ by 
\begin{equation}\label{eq:def_omega_t}
	\tilde \Omega^T L \tilde \Omega = L \quad\mathrm{with}\quad	\tilde \Omega = 
	\pmatrix{ \Omega & 0 \cr
	0 & \Omega^*}\,.
\end{equation}
The full symmetry $\tilde \Omega$ is unitary and symplectic with $\tilde
\Omega^T J \tilde \Omega = J$.  It also commutes with the hermiticity $[\tilde
\Omega, A] =0$.  For the blocks of $L$ it implies 
\begin{equation}
	\Omega^\dag H_\mathrm{eff} \Omega = H_\mathrm{eff} 
	\qquad\mathrm{and}\qquad \Omega N \Omega^T=N \,,
\end{equation}
i.e.  the symmetry acts on the effective Hamiltonian $H_\mathrm{eff}$ as the
defining symmetry on $H$.

The requirements  $ \Omega^\dag  \,iZ\, \Omega =iZ$ and $X \Omega^* X =
\Omega$ demand that it is of the form
\begin{equation}
	\Omega = \pmatrix{ P & 0 \cr 0 & P^* }
\end{equation}
with $P$ a unitary matrix. Thus a  symmetry of the Liouvillian is present if
the blocks  of $H_\mathrm{eff}$ and $N$ fulfill
\begin{eqnarray}
	& P^\dag h P =h, \qquad\;\;  P^T \Delta P=\Delta , \qquad P \bm z =
	\bm z, \cr 
	&  P^T U	P =U, \qquad P^\dag W P =W, \qquad  P^T V P^*=V\,.
\end{eqnarray}
We will discuss an
example of a unitary symmetry in section~\ref{ex1}.

\subsection{$\mathcal{PT}$-`symmetry'}

$\mathcal{PT}$-`symmetry' has been proposed as an alternative to the
conventional unitary quantum mechanics \cite{bender:99}. In the context of
open systems, $\mathcal{PT}$-`symmetry' is realized when the absorption and the
emission counteract each other.  If the symmetry is unbroken, we have that all
$\lambda_i$ are purely imaginary.  More generally, the eigenvalues come in
$\lambda$ and $-\lambda^*$ pairs, see  \cite{bender:99} for closed
systems  and figure~\ref{fig:spectrum}.  In our formalism, a
$\mathcal{PT}$-`symmetry' $\Omega$ is realized for an
anti-unitary $\Omega$~\footnote{Similar to above, given a unitary $\Omega$, the
expression $\Omega X K$ is a anti-unitary and fulfills all the conditions of a
$\mathcal{PT}$-`symmetry'.}  with $\Omega^\dag H_\mathrm{eff} \Omega =
H_\mathrm{eff}$, $ \Omega^\dag iZ \Omega = -iZ$, and $X \Omega^* X = \Omega$.
The anti-linearity maps $\lambda \mapsto \lambda^*$ while the reversal of $iZ$
introduces a minus sign. Thus, we  obtain eigenvalues in pairs  $\lambda,
-\lambda^*$. The reversal of $iZ$ corresponds to an exchange of creation and
annihilation operators.  This already signals the balance of gain and loss
\footnote{It has been shown in~\cite{scheel:21} that $\mathcal{PT}$-symmetry
can be alternatively achieved by a time-dependent modulation of the loss rate without
the need of gain.},
see also below.

Similarly to  above, we define for any anti-unitary $\Omega$ that the
Liouvillian is $\mathcal{PT}$-`symmetric' if
\begin{equation}\label{eq:pt}
\tilde \Omega^T L \tilde\Omega
= L \quad\mathrm{with}\quad	\tilde \Omega = 
	\pmatrix{ \Omega & 0 \cr
	0 & -\Omega}\,.
\end{equation}
Note that $\Omega^T J \Omega = - J$ (interchange of the creation and the
annihilation operators) and $[\tilde \Omega, A] =0$.  The blocks of $\Omega$
are fixed as
\begin{equation}
	\Omega = \pmatrix{ P & 0 \cr
	0 & P} K
\end{equation}
with a real orthogonal and symmetric matrix $P=P^T$.  This connects to the 
conventional definition of $\mathcal{PT}$-`symmetry', where the time reversal is
implemented by $K$ and the parity operator by a unitary $P$.  The action on
the individual parts of the blocks in $H_\mathrm{eff}$ and $N$ are given by
\begin{eqnarray}
	&P h^* P=h, \qquad\;\;  P \Delta^* P = \Delta,\qquad P \bm z
	=-\bm z^*, \cr 
	&P U P =U^\dag, \qquad P W P=V \,.
\end{eqnarray}
Therefore, systems exhibiting a $\mathcal{PT}$-`symmetry' can be identified by
two properties. First, the Hamiltonian of the system is
$\mathcal{PT}$-`symmetric' by itself.  Second, the dissipative, non-unitary part
of the Lindblad equation is invariant for an interchange of the gain and loss
processes. Note that our definition of $\mathcal{PT}$-`symmetry' (\ref{eq:pt})
is a generalization of the one proposed in~\cite{nakanishi:22, huber:20}. In
the case that the $\mathcal{PT}$-`symmetry' is unbroken with $\Re(\lambda_i) =
0$, the eigenstates can be chosen such that $\tilde \Omega \bi s_i = \bi s_i
$. In the case of a broken symmetry, acting with $\tilde \Omega$ maps an
eigenstate with eigenvalue $\lambda_i$ to its partner with eigenvalue
$-\lambda_i^*$;  the proof follows equivalent to (\ref{eq:a_act}).  For an
example of a $\mathcal{PT}$-`symmetric' system, see section~\ref{ex2}.

\section{Applications}
\subsection{Detuned parametric oscillator with Jordan decomposition}\label{ex1}

As a first application of our general method, we present an example for a
1-boson system where Jordan blocks appear.  The Hamiltonian is given by
\begin{equation}
	\mathcal{H} = \frac{\Delta}{2} a^\dag a + \frac{i\epsilon}{4} \left({a^\dag}^2 - a^2\right).
\end{equation}
It  has a physical realization as the degenerate parametric oscillator in the
rotating frame with detuning $\Delta$ and driving strength $\epsilon>0$.  The
coupling to the environment, with the thermal Bose-Einstein occupation
$\bar{n}$, can be modeled by single photon loss with rate $\gamma$ given by the jump
operators
\begin{equation}
	\mathcal{J}_1 = \sqrt{\gamma ( \bar{n} +1)}\, a, \qquad \mathcal{J}_2
	= \sqrt{\gamma \bar{n}} \,a^\dag.
\end{equation}
The effective Hamiltonian and noise matrix read
\begin{equation}
H_\mathrm{eff} = \frac{1}{2} \pmatrix{\Delta - i \gamma & i\epsilon \cr
-i\epsilon & \Delta + i \gamma},
\quad
N = \pmatrix{ 0 & \gamma(2\bar{n}+1) \cr
\gamma(2\bar{n}+1) & 0},
\end{equation}
while $\bi z =0$.  For $\epsilon = \Delta$ the Liouvillian matrix $L$ is not
diagonalizable and therefore the Jordan decomposition is needed
\cite{kamenev:23}.  The symplectic transition matrix
\begin{equation}
S = \pmatrix{ 1 & -i  & \frac{(\epsilon + i \gamma)(2 \bar{n} +1)}{\gamma} & \frac{(\epsilon^2 + \gamma^2)(2
	\bar{n} + 1)}{\gamma^2}  \cr
  0 & 1  & \frac{(i \epsilon + \gamma)(2 \bar{n} +1)}{\gamma} & \frac{\epsilon(i \epsilon + \gamma)(2 \bar{n} +
1)}{\gamma^2} \cr
0 & 0 & 1& 0 \cr
0 & 0 & i& 1
}.
\end{equation}
brings the Liouvillian  onto the form $\mathcal{L} = -\frac{\gamma}{2} (v_1u_1
+ v_2u_2) + \frac{\epsilon}{2} v_2u_1$.  Note that the transition matrix
$S$ fulfills the condition of trace preservation, meaning that the lower-left
block vanishes. The eigenstates are given by $\ket{0,n_2}$ with the
eigenvalues $\lambda = - \frac\gamma2 n_2$. The time evolution is given by
(\ref{eq:time_evolv_jordan}) with $\mu = -\frac\gamma2$ and $\nu =
\frac\epsilon2$.

For this model, we identify two symmetries.  First, the model has a
U(1)-symmetry at $\epsilon=0$ which corresponds to the bosonic number
conservation.  This symmetry is implemented by $P= e^{i\phi}$ which corresponds to
\begin{equation}
\Omega = \pmatrix{ e^{i\phi} & 0 \cr
0 & e^{-i\phi}}, \qquad \phi \in \mathbb{R}\,.
\end{equation}
For parametric driving with $\epsilon > 0$, the U(1)-symmetry is broken down  to
a $\mathbb{Z}_2$-symmetry ($\phi = 0,\pi$) with the remaining symmetry
operations $\Omega = \pm I_{2}$. The latter symmetry is always present when
$\bm z=0$.  It corresponds to the conservation of the parity of the number of
bosons.

\subsection{$\mathcal{PT}$-`symmetry' in coupled harmonic oscillators}\label{ex2}

We discuss an example for a Liouvillian $\mathcal{PT}$-`symmetric' system
for two bosonic modes ($m=2$).  The Hamiltonian
is given by, see \cite{nakanishi:22, roccati:21},
\begin{equation}
\mathcal{H} = \frac{g}{2}(ab^\dag + a^\dag b).
\end{equation}
It describes two coupled oscillators exchanging excitations with a coupling strength $g$.
The dissipation is modelled by 
\begin{equation}
	\mathcal{J}_1 = \sqrt{\gamma_l} \, a, \qquad \mathcal{J}_2 =
	\sqrt{\gamma_g} \, b^\dag\,,
\end{equation}
where $\gamma_l$ describes the loss rate of mode $a$ and $\gamma_g$  the
 gain rate of the  mode $b$.  The relevant matrices are given by
\begin{equation}
  H_\mathrm{eff}  = 
  \pmatrix{h_\mathrm{eff} & 0 \cr
  0 & h_\mathrm{eff}^*}, \qquad N =   \pmatrix{0 & n \cr
  n& 0}
\end{equation}
with
\begin{equation}
h_\mathrm{eff} =  \frac12   \pmatrix{-i\gamma_l  & g  \cr
g & i \gamma_g }, \qquad n = \pmatrix{\gamma_l & 0 \cr
0 &  \gamma_g }\,.
\end{equation}
At the point $\gamma_l = \gamma_g = \gamma$ where the gain matches the loss,
the system features a $\mathcal{PT}$-`symmetry' with $P=X$. The
corresponding anti-unitary symmetry matrix is given by
\begin{equation}
	\Omega = \pmatrix{ X & 0 \cr 0 & X } K\,.
\end{equation}
At the
$\mathcal{PT}$-`symmetric' point, the spectrum is given by $\lambda = \pm i
\sqrt{g^2 - \gamma^2}/2$. For $\gamma<|g|$, the $\mathcal{PT}$-`symmetry' is
unbroken. Still, due to the Jordan blocks, the time-evolution grows linearly
in time, see (\ref{eq:time_evolv_jordan}). As a result, the system is
(polynomially) unstable~\cite{albert:14}. For larger dissipation $\gamma >
|g|$, the $\mathcal{PT}$-`symmetry' is broken and the time-evolution diverges
exponentially.

\section{Conclusion}

In conclusion, we have shown how to diagonalize a complex symmetric matrix by
a symplectic transformation.  Using the superoperator formalism, we have
mapped the diagonalization of the Lindblad master equation onto this problem.
We have shown that the trace preservation of the Liouvillian is a crucial
ingredient for the method of `third quantization', allowing to define a Fock
space for the spectrum of the Liouvillian.  Effective non-Hermitian
Hamiltonians appear naturally in the approach as the off-diagonal blocks of
the Liouvillian.  They contain the details about the interaction of the
system with the environment as well as the (symplectic) eigenvalues.  An
additional noise matrix has been identified which only depends on the
dissipation.  It does not influence the spectrum and is only important for the
form of the eigenvectors, see \cite{prosen:10,clerk:23}.  

We have discussed the `diagonalization' of exceptional points for the simplest
case of $m=1$ boson.  In particular, we have given an explicit expression for
the time evolution operator.  We have shown how to include counting fields
(source terms) in the description, which allow to efficiently calculate   the
(factorial) cumulant generating function in the long-time limit.  Furthermore,
we have discussed symmetries of the Lindblad equation.  We have expressed the
hermiticity requirement as a `symmetry'. It explains the fact that the
eigenvalues are either real or appear in complex conjugate pairs.  It also
imposes a restriction on the Liouvillian, on observables, as well as on other
symmetries.  We have shown how symmetries of a closed system can be elevated
to a symmetry of the full Liouvillian. Moreover, we have investigated
$\mathcal{PT}$-`symmetries' which express the balance of gain and loss and
therefore describes systems on the verge of instability.  

\ack We acknowledge Yuli Nazarov for the initial motivation for the project and
fruitful discussions with Lisa Arndt. This work was supported by the Deutsche
Forschungsgemeinschaft (DFG) under Grant No.~HA 7084/8-1.

\appendix

\section*{Appendix}

As the symplectic diagonalization is not standard, we provide a basic Python3
script that, given a symmetric matrix $L$, returns the symplectic
transformation $S$ as well as the eigenvalues $\lambda$.
\begin{lstlisting}[language=Python,frame=single]
import scipy.linalg as la
import numpy as np

# Algorithm to transform L in the symplectic normal form
#
# Parameters: L :         symmetric matrix of even dimension
# Returns:    (lamb, S) : symplectic eigenvalues lamb 
#                          and transfer matrix S
#                          with S.T @ L @ S = [[0,lamb], [lamb,0]] 
#                          and S.T @ J @ S = J
def sympl_normal_form(L):
    n = L.shape[0] // 2   # dim of block 
    # standard symplectic form
    J = np.kron([[0,1], [-1,0]], np.identity(n))       
    d, Q = la.eig(L, J)   # solve the gen. eigenvalue problem
    # ind. of reord. eigenvalues (put pairs at correct place)
    ind = np.argsort(np.abs(d))[np.r_[:2*n:2, 1:2*n:2]] 
    d = d[ind]            # reorder the eigenvalues
    Q = Q[::,ind]         # reorder the eigenvectors
    T = Q @ J @ Q.T @ J   # calculate T (of theorem 1)
    B = la.sqrtm(T)       # calculate B (of theorem 1)
    return (d[:n], la.solve(B, Q))
\end{lstlisting}
Note that the program finds the pairs $\pm \lambda_i$ only by looking at the
absolute value of $\lambda_i$ and does not fix the normal form.  The sorting
works as long as there is no accidental degeneracy in $|\lambda_i|$. In the
case of $\mathcal{PT}$-`symmetry', when this assumption is not fulfilled, the
reordering part of the algorithm has to be adjusted.


\section*{Bibliography}


\begin{thebibliography}{10}

\bibitem{prosen:08}
T.~Prosen.
\newblock Third quantization: a general method to solve master equations for quadratic open {F}ermi systems.
\newblock {\em New J. Phys.}, 10:043026, 2008.

\bibitem{prosen:10}
T.~Prosen and T.~H. Seligman.
\newblock Quantization over boson operator spaces.
\newblock {\em J. Phys. A}, 43:392004, 2010.

\bibitem{ban}
M. Ban.
\newblock {L}ie-algebra methods in quantum optics: {T}he {L}iouville-space formulation.
\newblock {\em Phys. Rev. A}, 47:5093, 1993.

\bibitem{arimitsu}
T.~Arimitsu and H.~Umezawa.
\newblock {A} {G}eneral {F}ormulation of {N}onequilibrium {T}hermo {F}ield {D}ynamics.
\newblock {\em Prog. Theor. Phys.}, 74:429, 1985.

\bibitem{guo:17}
C.~Guo and D.~Poletti.
\newblock Solutions for bosonic and fermionic dissipative quadratic open systems.
\newblock {\em Phys. Rev. A}, 95:052107, 2017.

\bibitem{mcdonald:22}
A.~McDonald and A.~A. Clerk.
\newblock Exact {S}olutions of {I}nteracting {D}issipative {S}ystems via {W}eak {S}ymmetries.
\newblock {\em Phys. Rev. Lett.}, 128:033602, 2022.

\bibitem{barthel:21}
T.~Barthel and Y.~Zhang.
\newblock Solving quasi-free and quadratic {L}indblad master equations for open fermionic and bosonic systems.
\newblock {\em J. Stat. Mech.}, 2022:113101, 2022.

\bibitem{kamenev:23}
F.~Thompson and A.~Kamenev.
\newblock {Field Theory of Many-Body Lindbladian Dynamics}.
\newblock {\em arXiv:2301.02953}, 2023.

\bibitem{clerk:23}
A.~McDonald and A.~A. Clerk.
\newblock Third quantization of open quantum systems: new dissipative symmetries and connections to phase-space and {K}eldysh field theory formulations.
\newblock {\em arXiv:2302.14047}, 2023.

\bibitem{hassler:23}
F. Hassler, S. Kim, and L. Arndt.
\newblock {R}adiation statistics of a degenerate parametric oscillator at threshold.
\newblock {\em SciPost Phys.}, 14:156, 2023.

\bibitem{gantmacher:60}
F.~R. Gantmacher.
\newblock {\em Theory of Matrices}.
\newblock Chelsea Publishing Company, Vol.~1, Ch.~V, \S 1, 1959.

\bibitem{laub:74}
A.~J. Laub and K.~Meyer.
\newblock Canonical forms for symplectic and {H}amiltonian matrices.
\newblock {\em Celest. Mech. Dyn. Astron.}, 9:213, 1974.
\newblock See in particular corrolary 3.1.

\bibitem{nicacio:21}
F.~Nicacio.
\newblock Williamson theorem in classical, quantum, and statistical physics.
\newblock {\em Am. J. Phys.}, 89:1139, 2021.

\bibitem{minganti:18}
F.~Minganti, A.~Biella, N.~Bartolo, and C.~Ciuti.
\newblock Spectral theory of {L}iouvillians for dissipative phase transitions.
\newblock {\em Phys. Rev. A}, 98:042118, 2018.

\bibitem{kamenev:05}
A.~Kamenev.
\newblock {\em Field Theory of Non-Equilibrium Systems}.
\newblock Cambridge University Press, 2011.

\bibitem{colpa:78}
J.~H.~P. Colpa.
\newblock Diagonalization of the quadratic boson hamiltonian.
\newblock {\em Physica A}, 93:327, 1978.

\bibitem{bergholtz:21}
E.~J. Bergholtz, J.~C. Budich, and F.~K. Kunst.
\newblock Exceptional topology of non-Hermitian systems.
\newblock {\em Rev. Mod. Phys.}, 93:015005, 2021.

\bibitem{arndt:21}
L.~Arndt and F.~Hassler.
\newblock Universality of photon counting below a local bifurcation threshold.
\newblock {\em Phys. Rev. A}, 103:023506, 2021.

\bibitem{flindt:19}
F.~Brange, P.~Menczel, and C.~Flindt.
\newblock Photon counting statistics of a microwave cavity.
\newblock {\em Phys. Rev. B}, 99:085418, 2019.

\bibitem{padurariu:12}
C.~Padurariu, F.~Hassler, and {Yu}.~V. Nazarov.
\newblock Statistics of radiation at {J}osephson parametric resonance.
\newblock {\em Phys. Rev. B}, 86:054514, 2012.

\bibitem{arndt:19}
L.~Arndt and F.~Hassler.
\newblock Statistics of radiation due to non-degenerate {J}osephson parametric down conversion.
\newblock {\em Phys. Rev. B}, 100:014505, 2019.

\bibitem{mandel_wolf}
L.~Mandel and E.~Wolf.
\newblock {\em Optical Coherence and Quantum Optics}.
\newblock Cambridge University Press, 1995.

\bibitem{beenakker:98}
C.~W.~J. Beenakker.
\newblock Thermal radiation and amplified spontaneous emission from a random medium.
\newblock {\em Phys. Rev. Lett.}, 81:1829, 1998.

\bibitem{flindt:11}
D. Kambly, C. Flindt, and M. B\"uttiker.
\newblock Factorial cumulants reveal interactions in counting statistics.
\newblock {\em Phys. Rev. B}, 83:075432, 2011.

\bibitem{albert:14}
V.~V. Albert and L.~Jiang.
\newblock Symmetries and conserved quantities in {L}indblad master equations.
\newblock {\em Phys. Rev. A}, 89:022118, 2014.

\bibitem{buca:12}
B.~Buča and T.~Prosen.
\newblock A note on symmetry reductions of the {L}indblad equation: transport in constrained open spin chains.
\newblock {\em New J. Phys.}, 14:073007, 2012.

\bibitem{bender:99}
C.~M. Bender, S.~Boettcher, and P.~N. Meisinger.
\newblock $\mathcal{PT}$-symmetric quantum mechanics.
\newblock {\em J. Math. Phys}, 40:2201, 1999.

\bibitem{scheel:21}
L.~Teuber, F.~Morawetz, and S.~Scheel.
\newblock {P}assive $\mathcal{PT}$-symmetric {F}loquet coupler.
\newblock {\em Phys. Rev. A}, 103:063709, 2021.

\bibitem{nakanishi:22}
Y.~Nakanishi and T.~Sasamoto.
\newblock $\mathcal{PT}$ phase transition in open quantum systems with {L}indblad dynamics.
\newblock {\em Phys. Rev. A}, 105:022219, 2022.

\bibitem{huber:20}
J.~Huber, P.~Kirton, S.~Rotter, and P.~Rabl.
\newblock {Emergence of $\mathcal{PT}$-symmetry breaking in open quantum systems}.
\newblock {\em SciPost Phys.}, 9:052, 2020.

\bibitem{roccati:21}
F.~Roccati, S.~Lorenzo, G.~M. Palma, G.~T. Landi, M.~Brunelli, and F.~Ciccarello.
\newblock Quantum correlations in $\mathcal{PT}$-symmetric systems.
\newblock {\em Quantum Sci. Technol.}, 6:025005, 2021.

\end{thebibliography}
\end{document}